\def\ShowAuthNotes{0}

\documentclass[11pt]{article}
\usepackage{latexsym}
\usepackage{amssymb}
\usepackage{amsthm}
\usepackage{amsmath}
\usepackage{color}
\usepackage{tabularx}
\usepackage{url}
\definecolor{DarkBlue}{RGB}{0,0,150}
\usepackage[colorlinks,linkcolor=DarkBlue,citecolor=DarkBlue]{hyperref}

\usepackage{array}
\usepackage{epsfig}
\usepackage{fullpage}

\usepackage{bm}
\usepackage{bbm}

\usepackage{titlesec}
\titleformat{\section}{\large\bf}{\thesection}{1em}{}
\titleformat{\subsection}{\normalsize\bf}{\thesubsection}{1em}{}

\ifnum\ShowAuthNotes=1
\newcommand{\authnote}[3]{\textcolor{#3}{[{\footnotesize {\bf #1:} { {#2}}}]}}
\else
\newcommand{\authnote}[3]{}
\fi

\newcommand{\PRF}{\mathsf{PRF}}
\newcommand{\PRS}{\mathsf{PRS}}
\newcommand{\Keys}{\mathcal{K}}

\newtheorem{observation}{Observation}
\newtheorem{remark}{Remark}

\newtheorem{theorem}{Theorem}

\newtheorem{proposition}[theorem]{Proposition}
\newtheorem{definition}{Definition}
\newtheorem{claim}[theorem]{Claim}
\newtheorem{lemma}[theorem]{Lemma}
\newtheorem{conjecture}{Conjecture}
\newtheorem{notation}[definition]{Notation}
\newtheorem{corollary}[theorem]{Corollary}
\newtheorem{fact}[theorem]{Fact}

\def\bbC{{\mathbb C}}
\def\bbE{{\mathbb E}}
\def\bbF{{\mathbb F}}

\def\bbN{{\mathbb N}}

\def\bbR{{\mathbb R}}

\def\HT{{\bf HT}}

\def\binset{\{0,1\}}

\newcommand{\abs}[1]{\left\vert {#1} \right\vert}
\newcommand{\norm}[1]{\left\| {#1} \right\|}
\newcommand{\td}{\mathrm{TD}}

\newcommand{\qprf}{\mathrm{QPRF}}
\newcommand{\prs}{\mathrm{PRS}}
\newcommand{\ars}{\mathrm{ARS}}

\def\poly{{\rm poly}}

\def\negl{{\rm negl}}

\newcommand{\ket}[1]{|{#1}\rangle}
\newcommand{\bra}[1]{\langle{#1}|}

\title{
	(Pseudo) Random Quantum States with Binary Phase
	\author{Zvika Brakerski\thanks{Weizmann Institute of Science, \texttt{zvika.brakerski@weizmann.ac.il}. Supported by the Israel Science Foundation (Grant No.\ 468/14), Binational Science Foundation (Grants No.\ 2016726, 2014276), and by the European Union Horizon 2020 Research and Innovation Program via ERC Project REACT (Grant 756482) and via Project PROMETHEUS (Grant 780701).} \and Omri Shmueli\thanks{Tel Aviv University, \texttt{omrishmueli@mail.tau.ac.il}. Supported by the Zevulun Hammer Scholarship from the Council for Higher Education in Israel, and by Israel Science Foundation Grant No. 18/484, and by Len Blavatnik and the Blavatnik Family Foundation.}}}
\date{\vspace{-5ex}}
\date{}

\begin{document}
	\maketitle

	\begin{abstract}
		We prove a quantum information-theoretic conjecture due to Ji, Liu and Song (CRYPTO 2018) which suggested that a uniform superposition with random \emph{binary} phase is statistically indistinguishable from a Haar random state. That is, any polynomial number of copies of the aforementioned state is within exponentially small trace distance from the same number of copies of a Haar random state.
		
		As a consequence, we get a provable elementary construction of \emph{pseudorandom} quantum states from post-quantum pseudorandom functions. Generating pseduorandom quantum states is desirable for physical applications as well as for computational tasks such as quantum money. 
		We observe that replacing the pseudorandom function with a $(2t)$-wise independent function (either in our construction or in previous work), results in an explicit construction for \emph{quantum state $t$-designs} for all $t$. In fact, we show that the circuit complexity (in terms of both circuit size and depth) of constructing $t$-designs is bounded by that of $(2t)$-wise independent functions.
		Explicitly, while in prior literature $t$-designs required linear depth (for $t > 2$), this observation shows that polylogarithmic depth suffices for all $t$.
		
		We note that our constructions yield pseudorandom states and state designs with only real-valued amplitudes, which was not previously known. Furthermore, generating these states require quantum circuit of restricted form: applying one layer of Hadamard gates, followed by a sequence of Toffoli gates. This structure may be useful for efficiency and simplicity of implementation.
	\end{abstract}

	\section{Introduction}
	Randomness is one of the most fundamental resources for computation, and is indispensable for algorithms, complexity theory and cryptography. It is also a foundational tool for science in general, for purposes of describing and modeling natural phenomena. 
	As our understanding of nature expands to quantum phenomena, the importance of understanding the uniform distribution over quantum states, and being able to sample from it, naturally emerges. 
	
	Quantum states can be described as unit vectors in a high-dimensional complex Hilbert space. Thus, a random quantum state is just a random unit vector on this abstract sphere. This distribution is also referred to as the Haar measure over quantum states. We note that this is a continuous distribution, even if the Hilbert space is finite dimensional (i.e.\ can be described by a finite number of qubits). Since quantum states cannot be duplicated, the ability to generate random quantum states refers to the ability to generate multiple copies of the same random state vector. (In fact, a single copy of a quantum random state is identical to a classical random state.)  Haar random quantum states have numerous computational and physical applications. The former includes optimal quantum communication channels \cite{lloyd1997capacity}, efficient quantum POVM measurements \cite{renes2004symmetric} which are in turn useful in quantum state tomography, and gate fidelity estimation \cite{dankert2009exact}. The latter includes constructing physical models of quantum thermalization \cite{popescu2006entanglement}.
	
	Since random states have infinitely long descriptions (and super-exponential even if restricting to some finite precision), there is extensive literature studying approximate notions and specifically the notion of $\epsilon$-approximate $t$-designs. These are distributions whose $t$-tensor (i.e.\ taking $t$ copies of a sample from this distribution) are $\epsilon$ indistinguishable from (a $t$-tensor of) Haar (using the standard notion for statistical indistinguishability known as trace distance). We adopt the standard asymptotic convention and require by default that $\epsilon$ is negligible in our ``security parameter'', which we associate with the logarithm of the dimension of the Hilbert space. In this work we focus on quantum states over $n$ qubits (i.e.\ $2^n$ dimensional Hilbert space), so we associate our security parameter with $n$. However, our methods are extendable to any finite-dimensional space (with efficient representation). There is extensive literature studying (approximate) designs with bounded $t$, which also carry physical significance, see e.g.\ \cite{ambainis2007quantum, dankert2009exact, harrow2009random, nakata2013diagonal, nakata2014generating, kueng2015qubit}. Indeed, it is possible to efficiently generate $t$-designs using quantum circuits of size $\poly(t,n)$. Up to asymptotics, this matches the information theoretic bound (however, the important aspect of the depth complexity of generating $t$-designs remained open, to the best of our knowledge), and one cannot hope to efficiently generate $t$-designs for super-polynomial $t$.

	\paragraph{Asymptotically Random States, Pseudorandom States and the JLS Conjecture.}
	Ji, Liu and Song \cite{JLS18} (henceforth JLS) recently proposed to extend the notion of approximate designs. They proposed the notion of a \emph{pseudorandom quantum state} (PRS) which has a finite description but is \emph{computationally indistinguishable} from Haar given a $t$-tuple, for \emph{any} $t=\poly(n)$. Thus, for any computationally bounded purpose (experiment, naturally occurring process) a PRS is indistinguishable from a Haar state, regardless of the number of copies. They also showed that PRS are useful for cryptographic applications such as quantum money.
	
	Furthermore, \cite{JLS18} proposed an insightful template for constructing PRS. They start by showing that given quantum RAM access to exponentially many classical random bits, it is possible to construct a $\negl(n)$-approximate $n^{\omega(1)}$-design. Let us call such a distribution ARS, for Asymptotically Random State.\footnote{Actually, their ARS, as well as the one proven in this work, is even stronger: they show that for all $t$, their distribution is $O(t^2)/2^n$-approximate $t$-design.} An ARS is a statistical notion of PRS which has asymptotic limitations but no computational restrictions. Then, replacing the exponential random string with a quantum-query-resistant classically-computable pseudorandom function (PRF), the PRS construction naturally follows from ARS. The existence of such PRF is implied by the existence of quantum secure one-way functions \cite{zhandry2012construct}.
	
	The ARS construction of JLS is quite straightforward to describe. Generate a uniform superposition over all strings $x \in \binset^n$. This is described in the standard Dirac notation as $\sum_x \ket{x}$ (with some normalization factor). Then, assign a random quantum phase to each component $x$, i.e.\ generate $\sum_x \alpha_x \ket{x}$ for random independent roots of unity $\alpha_x$. To cope with finite precision, $\alpha_x$ is taken to a finite but exponential resolution $\alpha_x = \omega_{2^n}^{f(x)}$, where $f: \binset^n \to [2^n]$ is a random function and $\omega_{2^n}$ is the $2^n$-th root of unity. Given RAM access to the truth table of $f$, this state can be efficiently computed using Quantum Fourier Transform (QFT) modulo $2^n$.
	
	JLS then conjecture (but were unable to prove) that a much simpler construction, where $\alpha_x=(-1)^{f(x)}$, should also imply ARS. That is, replacing the ``high-resolution'' random phase, by the simplest binary phase. While this is only one of a few conjectures made in that work, it is the only one relevant to our work and we thus refer to it simply as the JLS conjecture.
	\begin{conjecture}[\cite{JLS18}, restated]
		The distribution over $n$-qubit quantum states defined by 
		\begin{align*}
		2^{-n/2}\sum_{x\in \binset^n} (-1)^{f(x)}\ket{x}
		\end{align*} where $f: \binset^n \to \binset$ is a random function, is an ARS.
	\end{conjecture}
	
	To highlight the gap between the conjecture and the provable ARS construction of JLS, let us describe a crucial point in the analysis of JLS. The analysis is based on an equivalence relation between $t$-tuples of $n$-bit strings, which naturally arises from the expression for statistical distance from Haar. The tuples $(x_1, \ldots, x_t)$, $(y_1, \ldots, y_t)$ are equivalent if their histograms (i.e.\ the number of times each $n$-bit string appears) are equal modulo $2^n$. Since $t < 2^n$ this condition is equivalent to requiring that the tuples are permutations of each other, which makes it possible to analyze the equivalence classes of this relation and for the analysis to go through.
	
	In the binary setting, the equivalence relates tuples whose histograms are equal modulo $2$. Thus the equivalence classes can no longer be described simply as a set and all of its permutations, and they don't even have the same size anymore. This creates many additional terms in the so called density matrix of the state (which is a complex matrix of exponential dimensions $2^{tn}\times 2^{tn}$). In order to prove the conjecture, one will have to show that the effect of these exponentially many new terms on the spectrum of the matrix is negligible and there seems to be no straightforward handle for this analysis. We resolve this problem in this work.

	\paragraph{Our Results -- Proving the Conjecture.} We prove the JLS conjecture, in fact we prove that the binary ARS implied by the conjecture has comparable properties to the prior construction (that used complex phase).
	
	\begin{theorem}[Main Result]\label{main_Thm}
		The distribution over $n$-qubit quantum states defined by 
		\begin{align*}
		2^{-n/2}\sum_{x\in \binset^n} (-1)^{f(x)}\ket{x}
		\end{align*} where $f: \binset^n \to \binset$ is a random function, is a $\frac{4t^2}{2^n}$-approximate $t$-design for all $t$, and thus an ARS.
	\end{theorem}
	This result has various implications that we describe below. We furthermore hope that our techniques will be useful for analyzing similarly complicated quantum states.

	We make two additional observations that refer to the requirement from a function $f$ to be plugged into either our theorem or that of JLS in order to imply PRS and quantum $t$-designs.
	\begin{enumerate}
		\item If we wish to obtain a PRS, the requirement of using a full-fledged quantum secure PRF can be relaxed. In fact, it is sufficient to have a function $f$ that is indistinguishable from random while allowing only uniform superposition queries (as opposed to arbitrary superposition queries). This leads to a quantum notion which is somewhat analogous to the classical notion of weak pseudorandom functions \cite{naor1999synthesizers}, an object that can be of interest for independent investigation and possibly more efficient constructions than PRFs.
		
		\item If we only wish to obtain a $t$-design, it is sufficient to replace $f$ with a $(2t)$-wise independent function, using the fact  that given $t$-quantum-query access, a $(2t)$-wise independent function is perfectly indistinguishable from a completely random function \cite{zhandry2012construct}.
		
	\end{enumerate}
	
	\paragraph{Implications.} We find the JLS conjecture compelling from aesthetic, conceptual and perhaps even practical reasons. In terms of aesthetics, it is bothersome that one would need to go into exponentially fine-grained resolution on the phase in order to generate an ARS/PRS, being able to achieve the same parameters with a more coarse resolution (and as we show next without compromising on parameters) seems to be a more desirable state of affairs. Conceptually, the result shows that ARS, which is for all efficiently observable purposes identical to a Haar random state, can be generated using only real-valued phases. Recalling that the Haar distribution is defined over complex vectors, it is appears not obvious that it can be approximated for all observable purposes by real-valued vectors.
	
	In terms of computational complexity, our construction uses circuits with restricted structure known in the literature as $\HT$  \cite{nest2008classical}. Concretely, the circuit contains a single parallel layer of Hadamard gates, followed by a circuit of Toffoli gates.
	This model is considered fairly weak (note that Hadamard and Toffoli are not even universal for arbitrary quantum computation) and in particular $\HT$ circuits are weakly classically simulatable (i.e.\ any distribution samplable by an $\HT$ circuit followed by measurement is also classically samplable). Result shows that even such a restricted model of quantum computation is enough to approximate the Haar measure.
	
	Lastly, from a practical standpoint, replacing the function $f$ by an efficient quantum-resilient PRF yields a very simple construction of a PRS, requiring only an $\HT$ circuit with the same circuit size and depth (up to asymptotics) as that of the PRF.
	Prior provable PRS candidates do not enjoy this property and appear to require a more complicated implementation (that in particular seem to need performing the Quantum Fourier Transform modulo $2^n$, or a similar procedure) to allow for the high-resolution of complex phase.
	
	In the context of generating $t$-designs, using our aforementioned observation and replacing $f$ with a $(2t)$-wise independent function (in either our theorem or JLS) implies a $t$-design construction with circuit size $\poly(t,n)$ and depth $O(\log t \cdot \log n)$.
	We are not aware of prior constructions of $t$ designs with $o(n)$ depth for $t > 2$ in the literature.
	Moreover, the $t$-design construction which is implied by our result can be implemented by an $\HT$ circuit with the same circuit size and depth (up to asymptotics) as that of the $(2t)$-wise independent function.
	
	\paragraph{Proof High-Level Overview.} Formally speaking, the proof follows by bounding the $\ell_1$ spectral norm of the difference between the density matrix of $t$-copies of the state with binary phase, and the density matrix of $t$-copies of the state with $2^n$ roots of unity. However, one needs not know much about density matrices, it suffices to say that we have a complex Hermitian matrix of dimensions $2^{tn}\times 2^{tn}$, where the sum of all eigenvalues is $0$, and we want to bound the sum of all absolute values of eigenvalues. It is thus sufficient to consider only positive or only negative eigenvalues. 
	
	Each row of the matrix corresponds to a tuple $(x_1, \ldots, x_t) \in (\binset^n)^t$ and each column corresponds to a tuple $(y_1, \ldots, y_t) \in (\binset^n)^t$. The entry in location $(x_1, \ldots, x_t), (y_1, \ldots, y_t)$ is nonzero if the aforementioned ``histogram condition'' holds on the tuples.\footnote{Recall that the (modulo-$2$) histogram condition states that $(x_1, \ldots, x_t), (y_1, \ldots, y_t)$ are equivalent if for all $z$, the number of times $z$ appears in the first tuple and the number of times it appears in the second tuple have the same parity.} In a bit more detail, up to a global $2^{-tn}$ scaling factor, if the modulo-$2$ histogram condition holds but the modulo-$2^n$ condition (i.e.\ permutation) does not hold then the entry will be $1$, but if both hold then there is a cancellation and the entry will be $0$.
	
	We start by observing that the matrix can be decomposed into ``combinatorial blocks'', each representing an equivalence class of the histogram relation. We analyze the properties of these blocks. We then provide two structural lemmas that together imply the theorem:
	\begin{enumerate}
		\item 	We provide a non-trivial upper bound on the rank of the matrix. While it is tempting to disregard the cancellations and just count the number of nonzero blocks and their respective rank, this implies an upper bound that is too coarse. We must therefore carefully take into account the cancellations induced by permutations in order to obtain a usable bound.
		
		\item We provide an upper bound on the absolute value of each negative eigenvalue. We do this by computing the characteristic polynomial of the matrix (the polynomial whose roots are the eigenvalues), which amounts to a product of the characteristic polynomials of the blocks. Within each block we obtain a closed form formula for the characteristic polynomial and show that its root cannot exceed a bound that is determined by the cardinality of the respective equivalence class (properly normalized).
	\end{enumerate}
	
	Combining the two lemmas by multiplying the rank bound with the eigenvalue absolute value bound implies the theorem.

	\paragraph{Paper Organization.} We use standard quantum and cryptographic notations and definitions, essentially following \cite{JLS18}, see short summary in Section~\ref{sec:prelim}. Our construction is presented in Section~\ref{sec:construction} and proven in Section~\ref{sec:conjproof}.

	\section{Preliminaries}
	\label{sec:prelim}
	
	For $m \in \bbN$, we denote $[m] := \{1, \cdots, m \}$.
	For a natural number $N$, denote by $\omega_{N} := e^{\frac{2\pi i}{N}}$ the complex root of unity of order $N$.
	Also for $N$, denote by ${\cal S}(N)$ the set of unit vectors in $\bbC^{N}$, by ${\cal D}(N)$ the set of $N \times N$ density matrices over $\bbC$, and by ${\cal U}(N)$ the set of $N \times N$ unitary matrices over $\bbC$.
	Note that for $n \in \bbN$, ${\cal S}(2^n)$ is the set of $n$-qubit pure quantum states, ${\cal D}(2^n)$ is the set of $n$-qubit mixed states, and ${\cal U}(2^n)$ is the set of $n$-qubit unitaries.
	When we consider quantum algorithms, we usually think of them as a uniform family of quantum circuits.
	
	When we consider eigenvalues and singular values of matrices throughout this paper, we implicitly refer to eigenvalues and singular values that possibly repeat, e.g. $\lambda_1 \geq \lambda_2 \geq \cdots \geq \lambda_n$ for matrix with $n$, possibly identical eigenvalues.
	
	The trace distance, defined below, is a generalization of statistical distance to the quantum setting and represents the maximal distinguishing probability between quantum states.
	\begin{definition}[Trace distance]
		Let $\rho_1, \rho_2 \in {\cal D}(2^n)$ be two density matrices of $n$-qubit mixed states.
		The trace distance between them is
		$$
		\td( \rho_1, \rho_2 ) := \frac{1}{2} \norm{\rho_1 - \rho_2}_1 \enspace,
		$$
		where for a hermitian matrix $M$, $\norm{M}_1 = \sum_{i}|\lambda_i|$, where $\lambda_i$ are the eigenvalues of $M$.
	\end{definition}
	
	The following is a basic fact that shows that classical circuits are a subset of quantum circuits. Recall that the Toffoli gate implements the $3$-qubit unitary defined by $\ket{x,y,z}\to\ket{x,y,z\oplus xy}$.
	\begin{proposition}[Toffoli gate is universal for classical computation]
		Let $f: \binset^n \to \binset^m$ be a function and let $C$ be a classical circuit that computes $f$. Define the unitary $U_f: \ket{x}\ket{y} \to \ket{x}\ket{y \oplus f(x)}$. Then there exists a quantum circuit of size $O(|C|)$ consisting only of Toffoli gates that computes $U_f$ (possibly using auxiliary $\ket{0}$ qubits).
	\end{proposition}	
	
	HT circuits are quantum circuits of a restricted structure, defined as follows.
	\begin{definition}[HT Circuit]
		A quantum circuit $C$ is an $\HT$ circuit if the first layer of the circuit consists of only Hadamard gates on a subset of the qubits, and the rest of the circuit consists of only Toffoli gates.
	\end{definition}

	\subsection{Pseudorandom Functions and $k$-Wise Independent Functions}
	Here we define pseudorandom functions with quantum security (QPRFs).	
	\begin{definition}[Quantum-Secure Pseudorandom Function (QPRF)]\label{def:qprf}
		Let $\Keys = \{ \Keys_n \}_{n \in \bbN}$ be an efficiently samplable key distribution, and let $\PRF = \{ \PRF_n \}_{n \in \bbN}$, $\PRF_n : \Keys_n \times \binset^n \rightarrow \binset^n$ be an efficiently computable function.
		We say that $\PRF$ is a quantum-secure pseudorandom function if for every efficient non-uniform quantum algorithm $A$ that can make quantum queries there exists a negligible function $\negl(\cdot)$ s.t. for every $n \in \bbN$,
		$$
		\abs{\Pr_{k \gets \Keys_n}[A^{\PRF(k, \cdot) }() = 1] - \Pr_{f \gets (\binset^n)^{(\binset^n)}}[A^{f}() = 1]} \leq \negl(n) \enspace .
		$$
	\end{definition}
	In \cite{zhandry2012construct}, QPRFs were proved to exist under the assumption that post-quantum one-way functions exist. 
	
	We define $k$-wise independent functions are keyed functions s.t.\ when the key is sampled uniformly at random, then any $k$ different inputs to the function generate $k$-wise independent random variables.
	\begin{definition}[$k(n)$-Wise Independent Function]\label{def:t-wise}
		Let $k(n) : \bbN \rightarrow \bbN$ be a function, $\Keys = \{ \Keys_n \}_{n \in \bbN}$ be a key distribution, and let $f = \{ f_n \}_{n \in \bbN}$, $f_n : \Keys_n \times \binset^n \rightarrow \binset^n$ be a function. Thus, $f$ is a $k(n)$-wise independent function if for all $n$, for every distinct $k(n)$ input values $x_1, \cdots, x_{k(n)} \in \binset^n$,
		$$
		\forall y_1, \cdots, y_{k(n)} \in \binset^n : 
		\Pr_{s \gets \Keys_n}[f(s, x_1) = y_1 \land \cdots \land f(s, x_{k(n)}) = y_{k(n)}] = 2^{-n\cdot k(n)} \enspace .
		$$
	\end{definition}
	It is not a part of the standard definition, but it is usually the case that we consider $\Keys$ to be efficiently samplable and $f$ to be efficiently computable.
	
	\subsection{Quantum Randomness and Pseudorandomness}
	\subsubsection{The Haar Measure on Quantum States}
	Intuitively, the Haar measure on quantum states is the quantum analogue of the classical uniform distribution over bit strings.
	More precisely, the Haar measure is the uniform (continuous) probability distribution on quantum states. 
	Recall that an $n$-qubit quantum state can be viewed as a unit vector in $\bbC^{2^n}$, thus the Haar measure on $n$ qubits is the uniform distribution over all unit vectors in $\bbC^{2^n}$.
	
	Formally, the density matrix representing the quantum distribution of drawing a random Haar vector and outputting $t$ copies of it is given below.
	\begin{definition}[$n$-Qubits, $t$-Copy Random Haar State]
		Let $t, n \in \bbN$, we define the $n$-qubits $t$-copy random Haar mixed state to be
		$$
		\rho_{( t, n, H )} := \int_{\ket{\psi} \in {\cal S}(2^n)} ( \ket{\psi}\bra{\psi} )^{\otimes t}\mathrm{d}\mu(\psi) \enspace,
		$$
		where $\mu( \cdot )$ is the Haar measure on ${\cal S}(2^n)$.
	\end{definition}
	
	\subsubsection{Approximate Quantum State $t$-Designs}
	Approximate $t$-designs are quantum distributions that are approximately random when the number of output copies of the sampled state is restricted. The formal definition follows.
	
	\begin{definition}[$n$-Qubits, $\varepsilon$-Approximate State $t$-Design]
		Let $\varepsilon \in [0, 1], t \in \bbN$, and let $\mathcal{Q}$ be a quantum distribution over $n$-qubit states.
		We say that $\mathcal{Q}$ is an $\varepsilon$-approximate state $t$-design if
		$$
		\td\big( \bbE_{\ket{\psi} \gets \mathcal{Q}}[ (\ket{\psi} \bra{\psi})^{\otimes t} ], \; \rho_{(t, n, H)} \big) \leq \varepsilon \enspace .
		$$
	\end{definition}
	
	For the sake of completeness, we give a definition for quantum state $t$-design generators.
	\begin{definition}[$\varepsilon(n)$-Approximate State $t(n)$-Design Generator]
		Let $\varepsilon(n): \bbN \rightarrow [0, 1]$, $t(n): \bbN \rightarrow \bbN$ be functions.
		We say that a pair of quantum algorithms $(K, G)$ is an $\varepsilon(n)$-approximate state $t(n)$-design generator if the following holds:
		\begin{itemize}
			\item \textbf{Key Generation.}
			For all $n$, $K(1^n)$ always outputs a classical key $k$.
			\item \textbf{State Generation.}
			For all $n$ and for all $k$ in the image of $K(1^n)$, there exists an $n$-qubit pure state $\ket{\phi_k}$ s.t.\ $G(1^n, k)=\ket{\phi_k}$.
			\item \textbf{Approximate Quantum Randomness.}
			For all $n$, the distribution $\ket{\phi_k}_{k \gets K(1^n)}$ is an $n$-qubit, $\varepsilon(n)$-approximate state $t(n)$-design.
		\end{itemize}
	\end{definition}
	Note that we define the generator as two algorithms instead of one, to highlight the fact that a state that is sampled can be generated multiple times on demand.
	
	For the purposes of this work it is convenient to define the notion of Asymptotically Random States (ARS) as follows.
	\begin{definition}[Asymptotically Random State (ARS)]\label{def:ars}
		An Asymptotically Random State ($\ars$) is shorthand for an asymptotic sequence of $\negl(n)$-approximate $n^{\omega(1)}$-designs.
	\end{definition}
	
	\subsubsection{Quantum Pseudorandomness}
	The notion of pseudorandom quantum states was introduced in \cite{JLS18}, was shown to be implied by QPRFs, and is defined below.
	\begin{definition}[Pseudorandom Quantum State (PRS)]\label{def:prs}
		A pair of quantum polynomial-time algorithms $(K,G)$ is a Pseudorandom State Generator (PRS Generator) if the following holds:
		\begin{itemize}
			\item \textbf{Key Generation.}
			For all $n$, $K(1^n)$ always outputs a classical key $k$.
			\item \textbf{State Generation.}
			For all $n$ and for all $k$ in the image of $K(1^n)$, there exists an $n$-qubit pure state $\ket{\phi_k}$ s.t.\ $G(1^n, k)=\ket{\phi_k}$.
			\item \textbf{Security.}
			For any polynomial $t(\cdot)$ and a non-uniform efficient quantum algorithm $A$ there exists a negligible function $\negl(\cdot)$ such that for all $n \in \bbN$,
			$$
			\abs{\Pr_{k \gets K(1^n)}[A\big( \ket{\phi_k}^{\otimes t(n)} \big) = 1] -
				\Pr_{\ket{\psi} \gets \mu}[A\big(\ket{\psi}^{\otimes t(n)} \big) = 1]} \leq \negl(n) \enspace ,
			$$
			where $\mu$ is the Haar measure on $\mathcal{S}(2^{n})$.
		\end{itemize}
		If the above holds, we say that the ensemble $\PRS = \{\PRS_n\}_{n \in \bbN}$, where $\PRS_n$ is the distribution $\ket{\phi_k}_{k \gets K(1^n)}$, is a Pseudorandom Quantum State (PRS) which is generated by $(K,G)$.
	\end{definition}
	In the above definition, the number of qubits in the pseudorandom states can also be parameterized (i.e. $G(1^n, k)$ can output $m(n)$-qubit states and not necessarily $n$-qubit states), but in the current work we will ignore this.

	\newcommand{\gbin}{\mathsf{G}_{\mathrm{bin}}}
	
	\section{Construction}
	\label{sec:construction}
	
	The following construction will be the base of both our pseudorandom state and quantum state $t$-design constructions.
	\begin{definition}[Binary Phase State Generator for $F$]
		Let $\Keys = \{ \Keys_n \}_{n \in \bbN}$ be a key space and let $F = \{F_n\}_{n \in \bbN}$ be a keyed (boolean) function $F_n : \Keys_n \times \binset^n \to \binset$.
		$\gbin^{F}$ is the procedure that takes as input a $k \in \Keys_n$ and outputs the superposition
		$$
		\ket{\phi_k} := 2^{-n/2} \sum_{x \in \binset^n} (-1)^{F_k(x)}\ket{x}~.
		$$
	\end{definition}
	
	The following claim establishes that $\gbin^F$ is efficiently implementable when $F$ is. 
	\begin{claim}\label{generation_claim}
		If $F$ is computable by a classical circuit of size $s(n)$ and depth $d(n)$, then $\gbin^F$ is computable by an $\HT$ circuit of size $O(s(n))$ and depth $d(n) + 1$.
	\end{claim}
	\begin{proof}
		The algorithm of $\gbin^F$ will get as input a key $k$ and generate the state $\ket{+}^{\otimes n}\ket{-}$ by performing $(n + 1)$ Hadamard gates (in parallel) $H^{\otimes (n + 1)}$ on the ancillary classical state $\ket{0}^{\otimes n}\ket{1}$, then execute the $F_k$ circuit (which can be realized quantumly by Toffoli gates) on the state $\ket{+}^{\otimes n}\ket{-}$.
		After the execution of $F_k$, the state is 
		$$
		2^{-n/2} \sum_{x \in \binset^n}(-1)^{F_k(x)}\ket{x}\ket{-},
		$$
		thus by tracing out the last qubit we get the output state $\ket{\phi_k}$.
	\end{proof}
	We note that previous candidates required a more involved generation process which required applying quantum Fourier transform modulo $2^n$, or a similar procedure.
	
	\subsection{Our Pseudorandom Quantum State (PRS) Generator and its Properties}
	Recall the definition of a $\prs$ (see Definition~\ref{def:prs}) and of a $\qprf$ (Defintion~\ref{def:qprf}).
	We present our construction of a $\prs$ candidate with binary phase as follows.
	\begin{claim}
		If $F$ is a $\qprf$ then $\gbin^F$ (along with the key generation algorithm of $F$) is a secure $\prs$ generator.
	\end{claim}
	\begin{proof}
		First, it's clear that the key generation algorithm $K$ of our PRS is the key generation algorithm of $F$ (that for input $1^n$, samples $k \gets \Keys_n$), and that the state generation algorithm $G$ of our PRS is $\gbin^F$.
		
		Now, we argue that by the quantum-security of $F$, for any polynomial number of copies $t(n)$, the distribution $\ket{\phi_k}_{k \gets \Keys_n}$ is computationally indistinguishable (by quantum adversaries) from a random binary phase state, that is, the distribution over $n$-qubit quantum states defined by 
		\begin{align*}
		2^{-n/2}\sum_{x\in \binset^n} (-1)^{f(x)}\ket{x} \enspace ,
		\end{align*}
		where $f: \binset^n \to \binset$ is a truly random function.
		
		By Theorem~\ref{main_Thm}, a random binary phase state is an ARS (Definition~\ref{def:ars}), which in particular means that a random Haar state and a random binary phase state are computationally indistinguishable for any polynomial number of copies.
		By the the triangle inequality of computational indistinguishability, we deduce that for any polynomial number of copies, the quantum distribution $\ket{\phi_k}_{k \gets \Keys_n}$ and the Haar distribution are computationally indistinguishable, which completes our proof.
	\end{proof}
	
	\begin{remark}
		We note that in our security proof we did not use the full power of quantumly secure PRFs. Indeed, it is sufficient to construct a PRF whose security only holds with respect to uniform superposition queries $\ket{+}^n$. This can be thought of as a quantum analog of the classical notion of weak PRFs \cite{naor1999synthesizers}. In the classical setting, it is conjectured that weak PRFs reside in a lower complexity class than full fledged PRFs \cite{akavia2014candidate}. If similar behavior can be shown in the quantum case it could improve the efficiency of PRS constructions.
		
		We leave the investigation of this new notion (which we propose to call quantumly weak PRFs) to future works.
	\end{remark}
	
	We conclude with observing that by our result, the complexity of PRSs is no greater than that of QPRFs, and is moreover implementable by $\HT$ circuits.
	\begin{corollary}
		Let $\PRF = \{ \PRF_n \}_{n \in \bbN}$ be a $\qprf$.
		Thus there is a $\prs$ generator construction $(K, G)$ implemented by $\HT$ circuits, where $K$ is implemented by circuits of the same size and depth as that of the key sampling algorithm of $\PRF$, and $G$ is implemented by circuits of the same size and depth (up to asymptotics) as that of $\PRF$.
	\end{corollary}
	
	\subsection{Shallow-Circuit Approximate $t$-Design Generators}
	We note that by a simple observation, we can replace the truly random function $f$ in Theorem \ref{main_Thm} with a $2t$-wise independent function to gain an elementary and efficient construction of quantum state approximate $t$-designs.
	Formally, we use the following fact.
	\begin{fact}[\cite{zhandry2012construct}, Fact 2]
		The behavior of any quantum algorithm making at most $q$ queries to a $2q$-wise independent function is identical to its behavior when the queries are made to a random function.
	\end{fact}
	This implies that when $f$ is a $2t$-wise independent function, then the state from Theorem \ref{main_Thm} is a $\frac{4t^2}{2^n}$-approximate $t$-design.
	We note that this observation can also be applied to the ARS from \cite{JLS18}, and it would imply a different (but seemingly less efficient) construction of $t$-designs.
	
	\begin{corollary}
		The distribution over $n$-qubit quantum states defined by 
		\begin{align*}
		2^{-n/2}\sum_{x\in \binset^n} (-1)^{f(x)}\ket{x}
		\end{align*} where $f: \binset^n \to \binset$ is a $2t$-wise independent function, is a $\frac{4t^2}{2^n}$-approximate $t$-design.
	\end{corollary}
	
	More explicitly, combining the above with Claim~\ref{generation_claim} implies that that when $f$ is a $2t$-wise independent function, $\gbin^{f}$ is an approximate $t$-design generator (along with the key generation algorithm of $f$). The following corollary relates the complexity of $t$-design generators with that of the $2t$-wise independent functions.
	
	\begin{corollary}
		Let $t(n) : \bbN \rightarrow \bbN$ be a function and let $f = \{ f_n \}_{n \in \bbN}$, $f_n : \Keys_n \times \binset^n \rightarrow \binset$ be a $(2t(n))$-wise independent function.
		Thus there is an $\frac{O(t(n)^2)}{2^n}$-approximate quantum state $t(n)$-design generator $(K, G)$ implemented by $\HT$ circuits, where $K$ is implemented by circuits of the same size and depth as that of the key sampling algorithm of $f$, and $G$ is implemented by circuits of the same size and depth (up to asymptotics) as that of $f$.
	\end{corollary}
	
	Finally, we can instantiate with known construction of $k$-wise independent functions to obtain the following.
	\begin{corollary}
		For every function $t(n) : \bbN \rightarrow \bbN$, there exists a $\frac{O(t(n)^2)}{2^n}$-approximate quantum state $t(n)$-design generator, implemented by $\HT$ circuits of $\poly(t(n), n)$ size and $O(\log t(n)\cdot\log n )$ depth.
	\end{corollary}
	\begin{proof}
		We recall the most elementary construction of $k$-wise independent distributions over $2^n$ variables. Consider the field $\bbF = \bbF_{2^n}$ and recall that $\bbF$ elements correspond to degree $(n-1)$ formal polynomials with binary coefficients. Thus there is a natural bijection between $\bbF$ and $\binset^n$ that allows to represent $\bbF$ elements as elements in $\binset^n$. This representation allows to perform field arithmetic operations using circuits of size $\poly(n)$ and depth $O(\log n)$.
		
		A $k$-wise independent distribution over $\bbF^\bbF$ is defined by the evaluations of a random degree $(k-1)$ polynomial over $\bbF$, on all elements in $\bbF$. The computational complexity of evaluating such a polynomial is $\poly(k,n)$ and its depth is $O(\log k \cdot \log n)$. Plugging in $k=2t$ completes the proof (note that we only require $2t$-wise independence over $\binset^\bbF$ so our instantiation is actually a slight overkill).
	\end{proof}

	\newpage
	\section{Proof of Theorem~\ref{main_Thm}}
	\label{sec:conjproof}
	We introduce the following notation.
	
	\begin{notation}[Complex phase state by $f$]
		For a function $f : \binset^n \rightarrow [ 2^n ]$ we denote
		$$
		\ket{f}_{(2^n)} := 2^{-n/2}\sum_{x \in \binset^n} \omega_{2^n}^{f(x)} \ket{x} \enspace .
		$$
		when it is clear from the context, the subscript $2^n$ will be dropped from $\ket{f}_{(2^n)}$.
	\end{notation}
	\begin{notation}[Binary phase state by $f$]
		For a function $f : \binset^n \rightarrow \binset$ we denote
		$$
		\ket{f}_{(2)} := 2^{-n/2}\sum_{x \in \binset^n} (-1)^{f(x)} \ket{x} \enspace .
		$$
		when it is clear from the context, the subscript $2$ will be dropped from $\ket{f}_{(2)}$.
	\end{notation}
	\begin{notation}[$t$-copy random complex phase mixed state]
		For $t, n \in \bbN$, denote
		$$
		\rho_{ ( t, n, 2^n ) } := \bbE_{f}[ (\ket{f}_{(2^n)}\bra{f}_{(2^n)})^{\otimes t} ] \enspace ,
		$$
		where the expectation is taken over a uniformly random function $f : \binset^n \rightarrow [ 2^n ]$.
	\end{notation}
	\begin{notation}[$t$-copy random binary phase mixed state]
		For $t, n \in \bbN$, denote
		$$
		\rho_{ ( t, n, 2 ) } := \bbE_{f}[ (\ket{f}_{(2)}\bra{f}_{(2)})^{\otimes t} ] \enspace ,
		$$
		where the expectation is taken over a uniformly random function $f : \binset^n \rightarrow \binset$.
	\end{notation}
	
	In \cite{JLS18} It is shown that the random complex phase state is an ARS.
	\begin{lemma}[\cite{JLS18}, Lemma 2]\label{Thm_H}
		Let $n, t \in \bbN$, then
		$$
		\td\big( \rho_{ ( t, n, 2^n ) }, \rho_{( t, n, H )} \big) = \prod_{i \in [t-1]}\Bigg( 1 - \frac{i}{2^n} \Bigg) - \prod_{i \in [t-1]}\Bigg( 1 - \frac{2\cdot i}{2^n + i} \Bigg) \enspace .
		$$
	\end{lemma}
	
	We will show that a random binary phase state is asymptotically statistically close to a random complex phase state. More precisely, we will prove the following lemma.
	\begin{lemma}\label{Th_1}
		Let $n, t \in \bbN$, then
		$$
		\td\big( \rho_{(t, n, 2)}, \rho_{(t, n, 2^n)} \big) \leq \prod_{i \in [t-1]}\Bigg( 1 + \frac{i}{2^n} \Bigg) - \prod_{i \in [t-1]}\Bigg( 1 - \frac{i}{2^n} \Bigg) \enspace .
		$$
	\end{lemma}
	
	Using the triangle inequality of trace distance and Lemmas \ref{Thm_H} and \ref{Th_1} (below, in the first inequality), we show that a random binary phase state is an ARS:
	\begin{gather*}
	\td\big( \rho_{ ( t, n, 2 ) }, \rho_{( t, n, H )} \big) \leq
	\prod_{i \in [t-1]}\Bigg( 1 + \frac{i}{2^n} \Bigg) - \prod_{i \in [t-1]}\Bigg( 1 - \frac{2\cdot i}{2^n + i} \Bigg) \leq \\
	\Bigg( 1 + \frac{t}{2^n} \Bigg)^t - \Bigg( 1 - \frac{2\cdot t}{2^n + t} \Bigg)^t \leq
	\Bigg( 1 + \frac{t}{2^n} \Bigg)^t - \Bigg( 1 - \frac{2\cdot t}{2^n} \Bigg)^t \underset{(*)}{\leq} \\
	1 + \frac{2\cdot t^2}{2^n} - \Bigg( 1 - \frac{2\cdot t}{2^n} \Bigg)^t \underset{(**)}{\leq} 
	1 + \frac{2\cdot t^2}{2^n} - \Bigg(1 - \frac{2\cdot t^2}{2^n} \Bigg) = 
	\frac{4\cdot t^2}{2^n} \enspace ,
	\end{gather*}
	where $(*)$ is due to one variant of Bernoulli's inequality ($\forall r > 1, x \in [0, \frac{1}{2(r - 1)}) : (1 + x)^r \leq 1 + 2rx$), and $(**)$ follows from the more popular variant of Bernoulli's inequality ($\forall r \notin (0, 1), x \geq -1 : (1 + x)^r \geq 1 + rx$).

	\bigskip
	
	\par\noindent Therefore, all that remains is to prove Lemma~\ref{Th_1}, which will require most technical effort.

	\subsection{Proof of Lemma \ref{Th_1}}
	
	Denote the difference matrix $\rho_n := \rho_{(t, n, 2)} - \rho_{(t, n, 2^n)}$. The proof of the lemma contains two main components.
	First, an upper bound on the number of non-zero eigenvalues of $\rho_n$.
	\begin{lemma}\label{Lem_2}
		Let $n \in \bbN$ and let $t \in \{ 1, 2, \cdots, 2^n - 1 \}$, thus the number of non-zero eigenvalues of $\rho_n$ is upper bounded by 
		$$
		{2^n + t - 1 \choose t} - {2^n \choose t} \enspace .
		$$
	\end{lemma}
	Second, a lower bound on the minimal (as in most negative) eigenvalue of $\rho_n$.
	\begin{lemma}\label{Lem_3}
		Let $n \in \bbN$ and let $t \in \{ 1, 2, \cdots, 2^n - 1 \}$, thus for all eigenvalues $\lambda$ of $\rho_n$ we have $-\frac{t!}{2^{tn}} \leq \lambda$.
	\end{lemma}
	\noindent Note that this will give an upper bound on the absolute values of all negative eigenvalues of $\rho_n$.
	
	Given the last two lemmas, we can prove Lemma \ref{Th_1}.
	
	\begin{proof}
		Let $n \in \bbN$ and let $t \in \{ 1, 2, \cdots, 2^n - 1 \}$.
		$\rho_n$ is a difference between two density matrices, and because that trace is linear and density matrices have a trace of 1, the trace of $\rho_n$ is 0.
		Also recall that the sum of eigenvalues of a matrix is equal to its trace, so, the positive and negative eigenvalues of $\rho_n$ balance each other to 0, and thus, a bound on the sum of absolute values of all eigenvalues of $\rho_n$ can be obtained by bounding the sum of the absolute values of its negative eigenvalues.
		Formally:
		$$
		\td\big( \rho_{(t, n, 2)}, \rho_{ ( t, n, 2^n ) } \big) = 
		\frac{1}{2}\norm{\rho_n}_1 = 
		\frac{1}{2} \cdot \sum_{\lambda \text{ eigenvalue of $\rho_n$}} |\lambda| = 
		\sum_{\lambda \text{ negative eigenvalue of $\rho_n$}} |\lambda| \enspace .
		$$
		
		Using Lemma \ref{Lem_2} and Lemma \ref{Lem_3}, we obtain an upper bound on the last sum, which yields the wanted inequality.
		
		\begin{gather*}
		\sum_{\lambda \text{ negative eigenvalue of $\rho_n$}} |\lambda| \leq
		\Bigg( {2^n + t - 1 \choose t} - {2^n \choose t} \Bigg)\frac{t!}{2^{tn}} = \\
		\frac{(2^n + t - 1)!}{2^{tn}(2^n - 1)!} - \frac{(2^n)! }{2^{tn} ( 2^n - t )!} = 
		\frac{(2^n)! \big( \prod_{i \in [t-1]}( 2^n + i ) \big)}{2^{tn}(2^n - 1)!} - \frac{(2^n)! \big( \prod_{i \in [t-1]}( 2^n - i ) \big)}{2^{tn}(2^n - 1)!} = \\	
		\frac{\prod_{i \in [t-1]}( 2^n + i )}{2^{(t - 1)n}} - \frac{\prod_{i \in [t-1]}( 2^n - i )}{2^{(t - 1)n}} = 
		\prod_{i \in [t-1]}\frac{ 2^n + i }{2^{n}} - \prod_{i \in [t-1]}\frac{ 2^n - i }{2^{n}} = \\
		\prod_{i \in [t-1]}\Bigg( 1 + \frac{i}{2^n} \Bigg) - \prod_{i \in [t-1]}\Bigg( 1 - \frac{i}{2^n} \Bigg) \enspace .
		\end{gather*}
	\end{proof}
	
	\subsection{The Structure of the Matrix $\rho_n$}
	
	We identify the structure of $\rho_n$ in order to prove Lemma \ref{rho_characterization}, which will be used in both proofs of Lemmas \ref{Lem_2}, \ref{Lem_3}.
	We do this by first describing $\rho_{( t, n, 2^n )}$ and $\rho_{( t, n, 2 )}$.
	More precisely, we will derive combinatorial expressions for $\rho_{( t, n, 2^n )}$ and $\rho_{( t, n, 2 )}$, and as a consequence we'll have an expression for their difference $\rho_n$.
	
	\subsubsection{The Structure of $\rho_{ ( t, n, 2^n ) }$}
	We will start with giving a formula for the entries of $\rho_{ ( t, n, 2^n ) }$; for convenience, the definition is restated:
	$$
	\rho_{( t, n, 2^n )} = 
	\bbE_{f \gets [ 2^n ]^{\binset^n}}\big[ ( \ket{f}\bra{f} )^{\otimes t} \big] = 
	\bbE_{f \gets [ 2^n ]^{\binset^n}}\big[ \ket{f}^{\otimes t}\bra{f}^{\otimes t} \big] \enspace .
	$$
	
	Observe that for a function $f : \binset^n \rightarrow [2^n]$, 
	$$
	\ket{f}^{\otimes t} = \Bigg( 2^{-n/2}\sum_{x \in \binset^n}\omega_{2^n}^{f(x)}\ket{x} \Bigg)^{\otimes t} = 
	2^{-tn/2}\sum_{\mathbf{x} = ( x_1, \cdots, x_t ) \in \binset^{n\times t}}\omega_{2^n}^{( \sum_{i \in [t]} f(x_i) )}\ket{\mathbf{x}} \enspace .
	$$
	
	Now we can compute $\rho_{( t, n, 2^n )}$:
	$$
	\rho_{( t, n, 2^n )} =
	\bbE_{f}\big[ \ket{f}^{\otimes t}\bra{f}^{\otimes t} \big] = 
	$$
	$$
	\bbE_{f}\Bigg[
	\Bigg( 2^{-tn/2}\sum_{\mathbf{x} = ( x_1, \cdots, x_t ) \in \binset^{n\times t}}\omega_{2^n}^{( \sum_{i \in [t]} f(x_i) )}\ket{\mathbf{x}} \Bigg) \cdot
	\Bigg( 2^{-tn/2}\sum_{\mathbf{y} = ( y_1, \cdots, y_t ) \in \binset^{n\times t}}\omega_{2^n}^{( -\sum_{i \in [t]} f(y_i) )}\bra{\mathbf{y}} \Bigg)
	\Bigg] = 
	$$
	$$
	2^{-tn}\sum_{\mathbf{x}, \mathbf{y} \in \binset^{n\times t}}\ket{\mathbf{x}}\bra{\mathbf{y}} \cdot \bbE_{f}\Bigg[ \omega_{2^n}^{( \sum_{i \in [t]} f(x_i) - \sum_{i \in [t]} f(y_i) )} \Bigg] \enspace , 
	$$
	So, for $\mathbf{x}, \mathbf{y} \in \binset^{n\times t}$, the $(\mathbf{x}, \mathbf{y})$-th entry of $\rho_{(t, n, 2^n)}$ is
	$$
	2^{-tn}\cdot \bbE_{f}\Bigg[ \omega_{2^n}^{( \sum_{i \in [t]} f(x_i) - \sum_{i \in [t]} f(y_i) )} \Bigg] \enspace .
	$$
	
	Now, define:
	\begin{definition}[$(t, n)$ permutations]
		Let $\mathbf{x}, \mathbf{y} \in \binset^{t \times n}$, and denote $\mathbf{x} = ( x_1, \cdots, x_t ), \mathbf{y} = ( y_1, \cdots, y_t )$, where $\forall i \in [t]: x_i, y_i \in \binset^n$.
		We say that $\mathbf{x}, \mathbf{y}$, are $( t, n )$ permutations of each other (or just permutations of each other) if there exists a permutation $\pi \in S_t$ s.t.
		$$
		( x_1, \cdots, x_t ) = ( y_{\pi(1)}, \cdots, y_{\pi(t)} ) \enspace .
		$$
	\end{definition}
	Note that an equivalent convenient characterization of the two strings $\mathbf{x}, \mathbf{y}$ being permutations of each other is that the multisets $\{ x_1, \cdots, x_t \}, \{ y_1, \cdots, y_t \}$ are equal.
	
	Observe that when $\mathbf{x}$ and $\mathbf{y}$ are permutations of each other, then for every $f$ we have $\sum_{i \in [t]} f(x_i) = \sum_{i \in [t]} f(y_i)$ and thus the expected value is 1 and the entry's value is $2^{-tn}$.
	We would like to also claim that if $\mathbf{x}, \mathbf{y}$ are not permutations of each other then the entry is 0, and it turns out we indeed can.
	Observe that if $\mathbf{x}, \mathbf{y}$ are not permutations of each other then there exists a string $s \in \binset^n$ that appears a different number of times in $\mathbf{x}$ and $\mathbf{y}$, and we can say that the $( \mathbf{x}, \mathbf{y} )$-th entry is
	$$
	2^{-tn}\cdot \bbE_{f}\Big[ \omega_{2^n}^{( \sum_{i \in [t]} f(x_i) - \sum_{i \in [t]} f(y_i) )} \Big] = 
	2^{-tn}\cdot \beta \cdot \bbE_{f}\Big[ \omega_{2^n}^{a \cdot f(s)} \Big] \enspace ,
	$$
	where $\beta\in \bbR$ is some real number (which we won't care about) and $a \in \{ -t, \cdots , -1, 1, \cdots, t \}$ is the (non-zero) difference between the number of appearances of $s$ in $\mathbf{x}$ and $\mathbf{y}$ (last equality follows from the fact that the expectation of a product of independent random variables is the product of expectations).
	Now we will use our restriction on $t$, which is that $t$ is strictly smaller then $2^n$.
	Combined with the fact that $a\neq 0$, it is necessarily the case that $\omega_{2^n}^{a} \neq 1$ (if $t$ could be as big as $2^n$ then $a$ will be able to be $2^n$ or some integer multiple of it, which will yield $\omega_{2^n}^{a} = 1$).
	After this restriction we obtain:
	$$
	\bbE_{f}\Big[ \omega_{2^n}^{a \cdot f(s)} \Big] =
	\sum_{i \in \{ 0, 1, \cdots, 2^n-1 \}}2^{-n}\cdot \omega_{2^n}^{a\cdot i} = 
	2^{-n}\cdot \Bigg( \frac{\omega_{2^n}^{a\cdot 2^n} - 1}{\omega_{2^n}^{a} - 1} \Bigg) = 0\enspace ,
	$$
	Finally, the above yields a combinatorial description of $\rho_{( t, n, 2^n )}$:
	$$
	\forall \mathbf{x}, \mathbf{y} \in \binset^{n \times t}: \rho_{( t, n, 2^n )}[ \mathbf{x}, \mathbf{y} ] = \begin{cases}
	2^{-tn} &\quad\enspace \mathbf{x}, \mathbf{y} \text{ are permutations}\\
	0 &\quad\enspace \mathbf{x}, \mathbf{y} \text{ are not permutations}\\
	\end{cases} \enspace .
	$$

	\subsubsection{The Structure of $\rho_{ ( t, n, 2 ) }$}
	By the same reasoning as in the case of $\rho_{ ( t, n, 2^n ) }$, we obtain that the $(\mathbf{x}, \mathbf{y})$-th entry of $\rho_{( t, n, 2 )}$ is
	$$
	2^{-tn}\cdot \bbE_{f}\Big[ (-1)^{( \sum_{i \in [t]} f(x_i) - \sum_{i \in [t]} f(y_i) )} \Big] \enspace ,
	$$
	where this time $f$ is a random function from $\binset^n$ to $\binset$ (rather than from $\binset^n$ to $[2^n]$).
	Because $( -1 ) = ( -1 )^{-1}$, the entry is simplified to 
	$$
	2^{-tn}\cdot \bbE_{f}\Big[ (-1)^{( \sum_{i \in [t]} f(x_i) + \sum_{i \in [t]} f(y_i) )} \Big] \enspace .
	$$
	
	Like in the case of $\rho_{( t, n, 2^n )}$, we would like a nice and clean combinatorial predicate to describe the entries of the matrix, and as we'll see in a bit, the matrix $\rho_{ ( t, n, 2 ) }$ indeed have the same general structure as $\rho_{ ( t, n, 2^n ) }$ but with different predicate on $\mathbf{x}, \mathbf{y}$.
	
	First, define the following:
	\begin{definition}[$(t, n)$ stabilizations]
		Let $\mathbf{x}, \mathbf{y} \in \binset^{t \times n}$, and denote $\mathbf{x} = ( x_1, \cdots, x_t ), \mathbf{y} = ( y_1, \cdots, y_t )$, where $\forall i \in [t]: x_i, y_i \in \binset^n$.
		We say that $\mathbf{x}, \mathbf{y}$, are $( t, n )$ stabilizations of each other (or just stabilizations of each other) if in the concatenated string $( \mathbf{x} \; \mathbf{y} ) = ( x_1, \cdots, x_t, y_1, \cdots, y_t )$, for every $s \in \binset^n$, $s$ appears an even number of times (this, of course, includes appearing 0 times).
	\end{definition}
	
	We note that the stabilization relation (which is all pairs that stabilize each other) is an equivalence relation over the set $\binset^{n \times t}$ (just like the permutation relation, which we didn't mention it being an equivalence relation, but it can easily be seen as one).
	It is clear that the stabilization relation is reflexive ($\mathbf{x}$ is always stabilizing $\mathbf{x}$), and it is also easy to verify that it is symmetric.
	To see why it is also transitive, we will use an additional characterization:
	\begin{definition}
		For a string $\mathbf{z} = ( z_1, \cdots, z_t )\in \binset^{n \times t}$ with $\forall i \in [t]: z_i \in \binset^n$, $\mathrm{Odd}( \mathbf{z} )$ is the set of strings from $\binset^n$ that appear an odd number of times in the sequence $( z_1, \cdots, z_t )$.
	\end{definition}
	
	\noindent For example, if $n = 3$ and $t = 7$ then $\mathrm{Odd}( 101, 111, 101, 000, 011, 111, 111 ) = \{ 111, 000, 011 \}$.
	
	We claim that two strings $\mathbf{x}, \mathbf{y}$ are stabilizations of each other if and only if $\mathrm{Odd}( \mathbf{x} ) = \mathrm{Odd}( \mathbf{y} )$.
	It is easy to verify the correctness of this claim, and also the fact that this claim implies the transitivity of the stabilization relation.
	
	To identify the elements of $\rho_{( t, n, 2 )}$ it remains to observe that when $\mathbf{x}, \mathbf{y}$ are stabilizations of each other then the entry is $2^{-tn}$, and when they are not, then we have $\mathrm{Odd}( \mathbf{x} ) \neq \mathrm{Odd}( \mathbf{y} )$ and it can be verified that the entry is 0, which yields the following description of $\rho_{ ( t, n, 2 ) }$:
	$$
	\forall \mathbf{x}, \mathbf{y} \in \binset^{n \times t}: \rho_{( t, n, 2 )}[ \mathbf{x}, \mathbf{y} ] = \begin{cases}
	2^{-tn} &\quad\enspace \mathbf{x}, \mathbf{y} \text{ are stabilizations}\\
	0 &\quad\enspace \mathbf{x}, \mathbf{y} \text{ are not stabilizations}\\
	\end{cases} \enspace .
	$$

	\subsubsection{Conclusion}
	
	Note that if $\mathbf{x}, \mathbf{y}$ are permutations then they necessarily stabilize each other, but the opposite is not true generally, furthermore, it is fairly easy to find stabilizing pairs that are not permutations, for instance $( 111, 000, 101, 101, 000 )$ and $( 110, 111, 111, 111, 110 )$.
	We'll call a pair of strings that suffice this demand (i.e. stabilize each other but are not permutations) remotely stabilized, that is:
	\begin{definition}[$(t, n)$ remote stabilizations]\label{remote_def}
		Let $\mathbf{x}, \mathbf{y} \in \binset^{t \times n}$, we say that $\mathbf{x}, \mathbf{y}$ are $( t, n )$ remote stabilizations of each other (or just remote stabilizations of each other) if they are stabilizations of each other but are not permutations of each other.
	\end{definition}
	
	In contrast to the cases of permutation and stabilization, remote stabilization is not an equivalence relation, and thus (generally speaking) it is harder to work with it.
	The stabilization relation is symmetric, but it is not reflexive, and in fact it is anti-reflexive, because a string is always a permutation of itself (and thus not a remote stabilization of itself), and it is also not transitive, because a (non-empty) relation which is symmetric and anti-reflexive can't be transitive.
	
	As said above, two strings that are permutations of each other are necessarily stabilizations of each other (in other words, the permutation relation is a refinement of the stabilization relation), and we deduce that $\rho_n$ has no negative terms and is also binary (scaled by $2^{-tn}$).
	Finally, this proves the characterization lemma of $\rho_n$.
	\begin{lemma}\label{rho_characterization}
		Let $n \in \bbN$ and let $t \in \{ 1, 2, \cdots, 2^n - 1 \}$, then the entries of $\rho_n$ can be given by the following formula:
		$$
		\forall \mathbf{x}, \mathbf{y} \in \binset^{n \times t}: \rho_{n}[ \mathbf{x}, \mathbf{y} ] = \begin{cases}
		2^{-tn} &\quad\enspace \mathbf{x}, \mathbf{y} \text{ are remote stabilizations}\\
		0 &\quad\enspace \mathbf{x}, \mathbf{y} \text{ are not remote stabilizations}\\
		\end{cases} \enspace .
		$$
	\end{lemma}

	\subsection{Proof of Lemma \ref{Lem_2}}
	\begin{proof}
		We will give an upper bound on the number of non-zero eigenvalues of $\rho_n$.
		$\rho_n$ is hermitian (and in particular diagonalizable) and thus the sum of dimensions of its eigenspaces sums up to the order of the matrix, which is $2^{tn}$.
		Also recall that the 0-eigenspace of $\rho_n$ is its kernel, thus by the rank-nullity theorem, the dimension of the 0-eigenspace plus the rank of $\rho_n$ equals the order of the matrix, $2^{tn}$.
		This means that the rank of $\rho_n$ equals the sum of dimensions of non-zero eigenspaces of $\rho_n$, which is exactly the number of non-zero (possibly identical) eigenvalues of $\rho_n$, thus, by giving an upper bound of $\text{rank}( \rho_n )$, we get an upper bound on the number of non-zero eigenvalues of $\rho_n$.
		
		It is a well known fact in linear algebra that elementary row operations does not change the rank of a matrix, it is also known that the rank of a matrix is bounded from above by the number of non-zero rows (the rank is the dimension of the row space, which in turn cannot be more than the number of non-zero rows), thus our bound on the rank of $\rho_n$ will come from looking at $\rho'_n$, a row-equivalent matrix to $\rho_n$, and bounding its number of non-zero rows.
		
		$\rho_n'$ is obtained by the following procedure:
		Recall that the permutation relation and the stabilization relation are both equivalence relations on $\binset^{t \times n}$ and thus induce equivalence classes.
		It will be useful (also for the proof of the next lemma) to define the following:
		\begin{definition}[Sentinel of an Equivalence Class]
			Let $C$ be an equivalence class of one of the two equivalence relations above (either the permutation relation or the stabilization relation).
			We define $\mathbf{x}_C \in \binset^{tn}$ the sentinel of $C$ to be the element in $C$ with the largest lexicographic order (where the lexicographic order of strings is as usual, with the most significant bit on the left, and least significant bit on the right).
		\end{definition}
		
		\begin{observation}
			Let $P$ be a permutation class of $\binset^{tn}$.
			Then, every pair in it $\mathbf{x}, \mathbf{y} \in P$ have the same set of remote stabilizers, and thus have identical rows in $\rho_n$.
		\end{observation}
		
		This means we can erase a bunch of redundant rows from $\rho_n$;
		for each permutation class $P$, take the sentinal row $\mathbf{x}_P$ of $P$ and subtract it from all other rows of strings from $P$.
		In the obtained matrix $\rho_n'$, the only non-zero rows are of sentinels.
		
		The number of sentinels is exactly the number of equivalence classes of the permutation relation, which in turn is the number of different multisets of $t$ elements from $\binset^n$ (note that a permutaion class can be defined by a multiset from $\binset^n$ of size $t$), and that number is known as common knowledge in combinatorics, usually referred to as "n multichoose k", in our case, ${2^n + t - 1 \choose t}$.
		
		\begin{observation}\label{Obs_2}
			Let $P$ be a permutation class of a multiset of $t$ distinct elements (essentially, a permutation class of a set of size $t$ with elements from $\binset^{n}$),
			then each of its elements have no remote stabilizers.
		\end{observation}
		
		The above observation basically says that strings of $t$ distinct elements are a special case where every stabilizer of them is also a permutation of them.
		This observation is useful to us because it means that for every permutation class $P$ of $t$ distinct elements, all rows of $P$ are zero-rows in the original $\rho_n$ (and thus so in $\rho_n'$).
		
		Furthermore, the reason that observation \ref{Obs_2} is important to our proof comes from the fact that there are ${2^n \choose t}$ such permutation classes, which is an overwhelming precentage from the total number of permutation classes ${2^n + t - 1}\choose{t}$.
		To conclude, we said that in $\rho_n'$, only the sentinels can possibly have non-zero rows, and that there are ${2^n + t - 1 \choose t}$ sentinels in total, but now we add the information that out of these ${2^n + t - 1 \choose t}$ sentinels, ${2^n \choose t}$ have zero rows, and thus, there are at most 
		$$
		{2^n + t - 1 \choose t} - {2^n \choose t}
		$$
		non-zero rows in $\rho_n'$ (and as a side note, there are in fact more zero-rows, for instance, for permutation classes of a multisets of the same element appearing $t$ times, but we won't care about these as their precentage is negligible).
		This concludes our proof of Lemma \ref{Lem_2}.
	\end{proof}

	\subsection{Proof of Lemma \ref{Lem_3}}
	We will give a lower bound on the most negative eigenvalue of $\rho_n$.
	Recall that $\rho_n$ is hermitian and thus has only real eigenvalues.
	Let $\lambda \in \bbR$, we know that $\lambda$ is an eigenvalue $\rho_n$ if and only if $\det( \rho_n - \lambda I ) = 0$.
	Denote by $\mathcal{A}$ the set of negative relative sizes of the permutation classes (along with 0),
	$$
	\mathcal{A} := \Biggl\{ -\frac{|P|}{2^{tn}} \enspace |\enspace P \text{ is a permutation class} \Biggl\} \; \cup \; \{ 0 \} \enspace , 
	$$
	where for a permutation class, its size is the number of different possible permutations of it, e.g. if $P$ is a permutation class of a multiset of the same element $t$ times, then $|P| = 1$, if $P$ is a permutation class of a multiset of $t$ distinct elements (in this specific case it is also a set) then $|P| = t!$, and if $P$ is a permutation class of a multiset of $(t - 2)$ distinct elements plus an additional distinct element that appears twice, then $|P| = {t \choose 2}\cdot (t - 2)!$ .
	We will show that there are no eigenvalues of $\rho_n$ smaller then all elements of $\mathcal{A}$, which will give us a lower bound of $( -\frac{t!}{2^{tn}} )$ on the minimal eigenvalue of $\rho_n$ (as this is the minimal element in $\mathcal{A}$).
	
	\paragraph{Determinant Calculation.}	
	Let $\lambda \in \bbR \setminus \mathcal{A}$, we will compute det$( \rho_n - \lambda I )$ by making the matrix $\rho_n - \lambda I$ lower triangular, and note that by assuming that $\lambda \notin \mathcal{A}$ we can, through the calculation of the determinant, apply row and column operations divided by expressions of the form $( \lambda + \frac{|P|}{2^{tn}} )$ for a permutation class $P$, or just by $\lambda$.
	
	Recall that the permutation relation is a refinement of the stabilization relation, which means that every stabilization class can be divided into a bunch of permutation classes, also recall that in the proof of Lemma \ref{Lem_2} we saw that for some stabilization classes, they are exactly a single permutation class and not a few (for example, according to the second observation in the proof, this is the case for permutation classes of sets of size $t$).
	We'll call such stabilization classes {\it trivial stabilization classes}, and because such stabilization classes are also permutation classes, we can say that these classes are also {\it trivial permutation classes}.
	By the second observation from the proof of Lemma \ref{Lem_2}, for a trivial stabilization class $S$, all of its rows are zero rows in $\rho_n$ and thus have only $-\lambda$ on their diagonal in the matrix $\rho_n - \lambda I$.
	The goal of the above reminder is to say that during the following determinant calculation, we will implicitly refer to {\it non-trivial stabilization classes} only.
	
	In order for the calculation of the determinant to be easy to follow, we will mention at each step $i$ the current operation we apply to the matrix, and then the value of the matrix after the operation.
	
	\begin{itemize}
		\item[-]
		\underline{\bf State 0:}
		This is just the matrix $\rho_n - \lambda I$.
		For a row $\mathbf{x} \in \binset^{n \times t}$, the matrix contains a $-\lambda$ on the diagonal, and $2^{-tn}$ in each of the columns of its remote stabilizers, all the rest is 0.
		Note that two rows $\mathbf{x}$ and $\mathbf{y}$ of the same permutation class differ exactly in the two columns $\mathbf{x}$ and $\mathbf{y}$.
		
		\item 
		\underline{\bf Step 1:}
		For each permutation class $P$, take the sentinel row $\mathbf{x}_P$ and subtract it from all rows of $P$.
		
		\item[-] 
		\underline{\bf State 1:}
		After step 1, for each permutation class $P$, $\mathbf{x}_P$ is unchanged, and for each non-sentinel row, it has (the unchanged) $-\lambda$ on the diagonal and $\lambda$ on the column of the sentinel, all the rest of its row is 0.
		
		\item 
		\underline{\bf Step 2:}
		For every permutation class $P$, take every column of a non-sentinel and add the column to the column of $\mathbf{x}_P$.
		
		\item[-]
		\underline{\bf State 2:}
		After this step, we introduced a more broad change to the matrix.
		For each permutation class $P$, each non-sentinel row is all zeros except a single $-\lambda$ on its diagonal, but the row of the sentinel $\mathbf{x}_P$ also changed:
		Note that the column operations we executed by non-sentinel columns of $P$ did not change the row of $\mathbf{x}_P$, but column operations by other permutation classes possibly did change $\mathbf{x}_P$.
		More precisely, the column operations made by the non-sentinels of other permutation classes that are in the same stabilization class with $\mathbf{x}_P$, did change $\mathbf{x}_P$.
		
		More specifically, now the row $\mathbf{x}_P$ has a $-\lambda$ on the diagonal, and for every permutation class $P' \neq P$ that is in the same stabilization class with $P$, all non-sentinel columns are unchanged (each entry is $2^{-tn}$) and the sentinel column $\mathbf{x}_{P'}$ has $\frac{|P'|}{2^{tn}}$.
		
		\item 
		\underline{\bf Step 3:}
		For every permutation class $P$, take each non-sentinel row $\mathbf{y}$, and for each permutation class $P' \neq P$ s.t. $P'$ is in the same stabilization class with $P$, add the scaled row $2^{-tn}\cdot\frac{1}{\lambda}\cdot\mathbf{y}$ to the row of the sentinel $\mathbf{x}_{P'}$.
		
		\item[-]
		\underline{\bf State 3:}
		Observe that because after last step all non-sentinels had only $-\lambda$ on their diagonal, what we actually do is adding $2^{-tn}\cdot\frac{1}{\lambda}\cdot (-\lambda) = - 2^{-tn}$ to an entry that had exactly $2^{-tn}$.
		After the third step, all non-sentinels of $P$ are unchanged (and hold zero everywhere except $-\lambda$ on the diagonal), and $\mathbf{x}_P$ now has $-\lambda$ on the diagonal, zero everywhere else except in the columns of other sentinels from the same stabilization class, and these columns are left unchanged from last step's ending (column of sentinel of permutation class $P'$ holds $\frac{|P'|}{2^{tn}}$).
		
		\item 
		\underline{\bf Step 4:}
		Let $S$ be a stabilization class with permutation classes $P_1, \cdots, P_k$, and consider the $k$ corresponding sentinels of these $k$ permutation classes; one of these sentinels is also the sentinel of the entire stabilization class, $\mathbf{x}_S$, and denote the permutation class holdoing $\mathbf{x}_S$ with $P_S$.
		Take the row of $\mathbf{x}_S$ and subtract it from the rows of all other sentinels of the permutation classes inside $S$.
		
		\item[-]
		\underline{\bf State 4:}
		After this step, for every stabilization class $S$, the row of $\mathbf{x}_S$ is unchanged and thus have $-\lambda$ on the diagonal, and for each permutation class $P$ from the same stabilization class $S$ s.t. $P \neq P_S$, it has $\frac{|P|}{2^{{tn}}}$.
		For every permutation class $P \neq P_S$ of $S$, $\mathbf{x}_{P}$ has $-( \lambda + \frac{|P|}{2^{tn}} )$ on its diagonal, $( \lambda + \frac{|P_S|}{2^{tn}} )$ on the column of $\mathbf{x}_S$, and 0 everywhere else.
		
		\item 
		\underline{\bf Step 5:}
		For each stabilization class $S$, for each permutation class $P$ of $S$ with $P \neq P_S$, add the scaled column 
		$$
		\Bigg(\frac{\lambda + \frac{|P_S|}{2^{tn}}}{\lambda + \frac{|P|}{2^{tn}}}\Bigg)\cdot
		\mathbf{x}_P
		$$
		to the column of $\mathbf{x}_S$.
		
		\item[-] 
		\underline{\bf State 5:}
		After the fifth (and final) step, the matrix is lower triangular (because we nullified all entries of the column of $\mathbf{x}_S$, and in the row of $\mathbf{x}_S$, everything to the right of the diagonal is 0, because $\mathbf{x}_S$ has maximal lexicographic ordering in $S$).
		First, observe that for trivial stabilization classes, all rows of the class are unchanged (we didn't operate on them during the calculation, and they were unaffected by column operations from other classes) and have all zeros except a $-\lambda$ on the diagonal.
		
		Let $S$ be a non-trivial stabilization class, the values of its rows are as follows:
		\begin{itemize}
			\item For every non-sentinel row, it has $-\lambda$ on the diagonal and all the rest is 0.
			\item For every permutation class $P$ of $S$ with $P \neq P_S$, the sentinel $\mathbf{x}_P$ of the permutation class has $-( \lambda + \frac{|P|}{2^{tn}} )$ on the diagonal and all the rest is 0.
			\item For the row of $\mathbf{x}_S$, denote by $P_1, \cdots, P_k$, $P_k = P_S$ the permutation classes of $S$.
			For every $i \in [k - 1]$, in column $\mathbf{x}_{P_i}$ of $\mathbf{x}_S$ it has $\frac{|P_i|}{2^{tn}}$ (but we actually won't care about this, only about the diagonal), and on the diagonal it has
			$$
			-\lambda + \sum_{i \in [k - 1]}\Bigg( \frac{\frac{|P_i|}{2^{tn}}}{( \lambda + \frac{|P_i|}{2^{tn}} )} \cdot \Big( \lambda + \frac{|P_S|}{2^{tn}} \Big) \Bigg) = 
			-\lambda + \Big( \lambda + \frac{|P_S|}{2^{tn}} \Big) \cdot \sum_{i \in [k - 1]}\Bigg( \frac{\frac{|P_i|}{2^{tn}}}{( \lambda + \frac{|P_i|}{2^{tn}} )} \Bigg) \enspace .
			$$
		\end{itemize}
	\end{itemize}

	Finally, from the values on the diagonal of the triangular matrix above, the determinant is 
	\begin{gather*}
	\forall \lambda \in ( \bbR \setminus \mathcal{A} ) : \det( \rho_n - \lambda\cdot I ) = 
	\prod_{S \text{ trivial stabilization class}}( -\lambda )^{|S|} \cdot \\
	\prod_{S \text{ non-trivial stabilization class}}
	\Big(\prod_{\; \mathbf{x} \text{ non-sentinel in $S$}}( -\lambda ) \Big) \cdot
	\Bigg(\prod_{P \text{ permutation class in $S$ with $P \neq P_S$}} \Big( - \lambda - \frac{|P|}{2^{tn}} \Big) \Bigg) \cdot \\
	\Bigg( -\lambda + \Big( \lambda + \frac{|P_S|}{2^{tn}} \Big) \cdot \sum_{P \text{ permutation class in $S$ with $P \neq P_S$}}\Bigg( \frac{\frac{|P|}{2^{tn}}}{( \lambda + \frac{|P|}{2^{tn}} )}  \Bigg) \Bigg)
	\enspace .
	\end{gather*}
	
	\paragraph{Using The Determinant to Show The Lower Bound.}
	Given the above determinant result, we can now finally prove Lemma \ref{Lem_3}.
	\begin{proof}
		Assume towards contradiction that there is a real number $\lambda' < -\frac{t!}{2^{tn}}$ (note that this implies $\lambda' \notin \mathcal{A}$) such that it is an eigenvalue of $\rho_n$, thus $\det( \rho_n - \lambda'\cdot I ) = 0$ and thus it is necessarily the case that one of the terms in the above product (of the determinant) has to be 0.
		Due to $\lambda' \notin \mathcal{A}$, it can be seen that the only terms that can possibly be 0 in the above product are the terms of the diagonals of the sentinels of (non-trivial) stabilization classes, so let's check what happens in these terms.
		
		Let $S$ be a non-trivial stabilization class, and consider the term of the sentinel $\mathbf{x}_S$ in the product above: 
		$$
		-\lambda' + \Bigg( \lambda' + \frac{|P_S|}{2^{tn}} \Bigg) \cdot \sum_{P \text{ permutation class in $S$ with $P \neq P_S$}} \Bigg( \frac{\frac{|P|}{2^{tn}}}{( \lambda' + \frac{|P|}{2^{tn}} )} \Bigg) \enspace .
		$$
		We have,
		$$
		\forall P \text{ permutation class in $S$, including $P_S$} : \lambda' < -\frac{|P|}{2^{tn}} \enspace ,
		$$
		and thus
		$$
		\Bigg( \lambda' + \frac{|P_S|}{2^{tn}} \Bigg) < 0, \enspace \sum_{P \text{ permutation class in $S$ with $P \neq P_S$}}\Bigg( \frac{\frac{|P|}{2^{tn}}}{( \lambda' + \frac{|P|}{2^{tn}} )}  \Bigg) < 0\enspace ,
		$$
		which implies
		$$
		\Bigg( \lambda' + \frac{|P_S|}{2^{tn}} \Bigg) \cdot \sum_{P \text{ permutation class in $S$ with $P \neq P_S$}}\Bigg( \frac{\frac{|P|}{2^{tn}}}{( \lambda' + \frac{|P|}{2^{tn}} )}  \Bigg) > 0\enspace .
		$$
		Finally, due to $-\lambda'$ being in particular positive, the term has to be positive as well:
		$$
		-\lambda' + \Bigg( \lambda' + \frac{|P_S|}{2^{tn}} \Bigg) \cdot \sum_{P \text{ permutation class in $S$ with $P \neq P_S$}}\Bigg( \frac{\frac{|P|}{2^{tn}}}{( \lambda' + \frac{|P|}{2^{tn}} )}  \Bigg) > 0 \enspace ,
		$$
		in contradiction to $\det( \rho_n - \lambda' I ) = 0$.
	\end{proof}

	\subsubsection*{Acknowledgments} We thank Henry Yuen and Vinod Vaikuntanathan for insightful discussions. In particular thanks to Henry for pointing us to the \cite{JLS18} result.
	We also thank Aram Harrow for providing advice regarding the state of the art.

	\bibliographystyle{alpha}
	\bibliography{PseudoRandomQuantum}

\end{document}